%% file: offload_gain-letter.tex
\documentclass[9pt,twocolumn,twoside]{IEEEtran}

\input{public_files/mypackage}

\input{public_files/newcommands}

\input{public_files/acronyms}

\begin{document}
\title{Optimizing Joint Probabilistic Caching and Channel Access for Clustered {D2D} Networks}
 

 \author{Ramy Amer,~\IEEEmembership{Student~Member,~IEEE,} M. Majid Butt,~\IEEEmembership{Senior,~IEEE,} 
~and~Nicola~Marchetti,~\IEEEmembership{Senior~Member,~IEEE}
 
\thanks{Ramy Amer and Nicola~Marchetti are with CONNECT Centre for Future Networks, Trinity College Dublin, Ireland. Email:\{ramyr, nicola.marchetti\}@tcd.ie.}
\thanks{M. Majid Butt is with Nokia Bell Labs, France and Trinity College Dublin, Ireland. Email: majid.butt@nokia-bell-labs.com.}
\thanks{This publication has emanated from research conducted with the financial support of Science Foundation Ireland (SFI) and is co-funded under the European Regional Development Fund under Grant Number 13/RC/2077.}
} 
\maketitle

\maketitle			
\begin{abstract} 
Caching at mobile devices and leveraging device-to-device (D2D) communication are two promising approaches to support massive content delivery over wireless networks. Analysis of such D2D caching networks based on a physical interference model is usually carried out by assuming uniformly distributed devices. However, this approach does not capture the notion of device clustering. In this regard, this paper proposes a joint communication and caching optimization framework for clustered D2D networks. Devices are spatially distributed into disjoint clusters and are assumed to have a surplus memory that is utilized to proactively cache files, following a random probabilistic caching scheme. The cache offloading gain is maximized by jointly optimizing channel access and caching scheme. A closed-form caching solution is obtained and bisection search method is adopted to heuristically obtain the optimal channel access probability. Results show significant improvement in the offloading gain reaching up to $10\%$ compared to the Zipf caching baseline.  
\end{abstract}
\vspace{-0.3 cm}
\section{Introduction}
Caching at mobile devices significantly improves system performance by facilitating \ac{D2D} communications, which enhances the spectrum efficiency and alleviate the heavy burden on backhaul links \cite{Femtocaching}. There are two main approaches for content placement in the literature, deterministic and probabilistic. For deterministic placement, files are cached and optimized for specific networks in a deterministic manner \cite{Femtocaching,golrezaei2014base,8412262}. However, in practice, the wireless channels and the geographic distribution of devices are time-variant. This triggers the optimal content placement strategy to be frequently updated, which makes the content placement quite complex. To cope with this problem, probabilistic content placement is proposed whereby each device randomly caches a subset of the content with a certain caching probability in stochastic networks \cite{chen2017probabilistic}. 
In this paper, we focus on the probabilistic content placement problem.

Modeling of wireless caching networks also follows two main directions in the current state-of-art. The first line of work focuses on the fundamental scaling results by assuming a simple protocol channel model \cite{Femtocaching,golrezaei2014base,8412262}, known as the protocol model.  This model assumes that two devices can always communicate if they are within a certain distance. The second line of work, which is similar to the one adopted in this paper, considers a more realistic model for the underlying physical layer \cite{andreev2015analyzing}. This is commonly defined as the physical interference model. 

The analysis of wireless caching networks that underlies a physical interference model, is commonly conducted by means of stochastic point processes. For instance, modeling device locations as a \ac{PPP}  is a widely adopted  approach in the wireless caching area \cite{andreev2015analyzing,amer2020caching}. However, while the PPP model is tractable, a realistic model for \ac{D2D} caching networks needs to capture the notion of clustering. In particular, in clustered \ac{D2D} networks, each device has multiple proximate devices, where any of them can act as a serving device. Such deployments can be characterized by cluster processes \cite{haenggi2012stochastic}.

The performance of clustered D2D caching networks is studied in \cite{clustered_twc} and \cite{8374852}. For instance, the authors in \cite{clustered_twc} discussed different strategies of content placement in a \ac{PCP} deployment. Moreover, the authors in \cite{8374852} proposed cooperation among the D2D transmitters and probabilistic caching strategies to save the energy cost of content providers, where the location of these providers is modeled by a Gauss-Poisson process. However, while the works in \cite{clustered_twc,8374852} studied clustered D2D networks from different perspectives, the joint optimization of caching and communication for clustered \ac{D2D} networks has not been addressed yet in the literature.

Compared with this prior art \cite{andreev2015analyzing,amer2020caching,haenggi2012stochastic,clustered_twc,8374852}, in this paper we study the content placement and delivery for a network wherein cache-enabled devices are spatially distributed into disjoint clusters. We conduct a performance analysis and joint optimization of channel access and probabilistic content placement aiming to maximize the cache offloading gain. We characterize the optimal content placement as a function of the system parameters, and propose a heuristic approach to obtain the optimal channel access probability. Our results reveal that the optimal caching scheme heavily depends on the channel access probability and the geometry of the network. \emph{Overall, joint optimization of content placement and communication, e.g., channel access, is shown to be vital to enhance the performance of wireless caching networks.}   

\vspace{-0.4 cm}
\section{System Model}
\subsection{System Setup}
We model the location of mobile devices with a \ac{TCP}. The \ac{TCP} is composed of the parent points, which are drawn from a \ac{PPP} $\Phi_p$ with density $\lambda_p$, and the daughter points that are drawn from a Gaussian \ac{PPP} around each parent point \cite{haenggi2012stochastic}. In particular, the daughter points  are normally scattered with variance $\sigma^2 \in \mathbb{R}$ around each parent point. The parent points and offspring are referred to as cluster centers and cluster members, respectively. By the \ac{TCP} definition, the number of devices per cluster is a Poisson \ac{RV} with mean $\overline{n}$. Therefore, the density function of a cluster member  location relative to its cluster center is
\begin{equation}
f_Y(y) = \frac{1}{2\pi\sigma^2}\textrm{exp}\Big(-\frac{\lVert y\rVert^2}{2\sigma^2}\Big),	\quad \quad y \in \mathbb{R}^2
\label{pcp}
\end{equation}
where $\lVert .\rVert$ is the Euclidean norm. The intensity function of a cluster is given by $\lambda_c(y) = \frac{\overline{n}}{2\pi\sigma^2}\textrm{exp}\big(-\frac{\lVert y\rVert^2}{2\sigma^2}\big)$, and therefore, the intensity of the entire process is given by $\lambda = \overline{n}\lambda_p$. 

We assume that the \ac{D2D} communication is operating as out-of-band D2D under flat Rayleigh fading channels.   \ac{D2D} communication is enabled within each cluster to deliver popular content. It is assumed that the devices adopt a slotted-ALOHA medium access protocol, where each transmitter during each time slot, independently and randomly accesses the channel with the same probability $q$. One can alternatively assume that each device makes a coin flip at each time about whether or not it accesses a shared-channel. This allows us to define a Bernoulli process $N_y$ with the probability that a device located at $y$ accesses a channel being $\Pb(N_y)=q$. The key advantage of adopting slotted-ALOHA is that it is a simple yet fundamental \ac{MAC} protocol, where there is no need for a central controller to schedule the users' transmissions. Moreover, despite the vast amount of existing studies on \ac{MAC} protocols, only variations of ALOHA and CSMA are still used in the majority of technologies being adopted for the Internet of Things \cite{7835337}. According to this access model, multiple active \ac{D2D} links might coexist within a cluster. Therefore, $q$ is a design parameter that directly controls intra- as well as inter-cluster interference, as described later.


If a requesting device  caches the desired content, the device directly retrieves the content. However, if  the content is not locally cached, it can be downloaded from a randomly selected neighboring device that caches the file within the same cluster, henceforth called \emph{catering device}. This catering device is, in turn, admitted to access the channel according to the proposed slotted-ALOHA protocol. Finally, the device attaches to the nearest \ac{BS} as a last resort to download the content, in the case it is not cached within the device cluster. 
Since there are memory and battery consumption costs borne by a catering device, the geographically closest device may not want to participate in the content caching and/or delivery. Hence, randomizing the catering device reflects the possibility of being served by a distant device that is willing to participate in the content delivery, while not necessarily being the nearest one. Note that this assumption is commonly adopted in the literature \cite{clustered_twc} and \cite{8374852}. 

\subsection{Content Popularity and Caching}
We assume that each device has a surplus memory of size $M$ designated for caching files. The total number of files is $N_f> M$, and the set (library) of content indices is denoted as $\mathcal{F} = \{1, 2, \dots , N_f\}$. These files represent the content catalog that all devices in a cluster may request, which are indexed in a descending order of popularity. The probability that the $i$-th file is requested follows a Zipf's distribution given by,
\begin{equation}
p_i = \frac{ i^{-\beta} }{\sum_{k=1}^{N_f}k^{-\beta}},
\label{zipf}
\end{equation}
where $\beta$ is a parameter that reflects how skewed the popularity distribution is. For example, if $\beta= 0$, the popularity of the files has a uniform distribution. Increasing $\beta$ increases the disparity among the files' popularity such that lower indexed files have higher popularity. By definition, $\sum_{i=1}^{N_f}p_i = 1$. 
We use the Zipf's distribution to model the popularity of files per cluster \cite{Delay-Analysis,8886101,amer2018optimizing,amer2019performance,amer20200caching,chaccour2019Reliability}.

We adopt a random content placement where each device independently and probabilistically selects a file to cache according to the function $\boldsymbol{b} = \{b_1, b_2, \dots, b_{N_{f}}\}$, where $b_i$ is the probability that a device caches the $i$-th file, $0 \leq b_i \leq 1$ for all $i\in\{1, \dots, N_f\}$. To avoid duplicate caching of the same content within the memory of a device, we follow a probabilistic caching approach with $\sum_{i=1}^{N_f}b_i=M$. 

%

Next, we proceed with the rate coverage analysis to obtain the offloading gain, which is a key performance metric for D2D caching networks \cite{cache_schedule}. Particularly, the offloading gain is defined as the probability of obtaining a requested file from the local cluster, either via self-cache or from a neighboring device in the same cluster, with a received \ac{SIR} higher than a required threshold $\vartheta$. 
\section{Rate Coverage Analysis} 
We conduct the next analysis for a cluster whose center is assumed at $x_0\in \Phi_p$, referred to as representative cluster. The device requesting a content in this cluster, henceforth called typical device, is located at the origin. We denote the location of the catering device by $y_0$ relative to $x_0$, where $x_0, y_0\in \mathbb{R}^2$. The distance between the typical and catering devices is denoted as $r=\lVert x_0+y_0\rVert$, which is a realization of a \ac{RV} $R$ whose distribution is described later. Having explained the channel access and the  random selection of catering devices, the offloading gain can be expressed as 
 \begin{align}
 \Pb_o(q,\boldsymbol{b}) &= \sum_{i=1}^{N_f} p_i b_i + p_i(1 - b_i)(1 - e^{-b_i\overline{n}}) \times
 \nonumber \\
 &
  \label{offloading_gain}
 \quad\quad\quad\quad \quad\quad \underbrace{\int_{r=0}^{\infty}f_R(r) \Pb(\sir_{|r}>\vartheta) \dd{r}}_{\Upsilon},
 \end{align}
where $\sir_{|r}$ is the received \ac{SIR} at the typical device when downloading a content from a catering device $r$ apart from the origin, and $\Upsilon$ represents the rate coverage probability. The first term in (\ref{offloading_gain}) is the probability of requesting a locally cached file (self-cache). The second term is the probability that a requested file $i$ is cached in at least one cluster member and being downloadable with an \ac{SIR} greater than $\vartheta$, given that it was not self-cached. More precisely, since the number of devices per cluster has a Poisson distribution, the probability that there are $k$ devices per cluster is equal to $\frac{\overline{n}^k e^{-\overline{n}}}{k!}$. Accordingly, the probability that there are $k$ devices caching content $i$ is $\frac{(b_i\overline{n})^k e^{-b_i\overline{n}}}{k!}$. Hence, the probability that at least one device caches content $i$ is $1 - e^{-b_i\overline{n}}$.

For the serving distance distribution $f_R(r)$, since both the typical device and catering device have their locations drawn from a normal distribution with variance $\sigma^2$, then by definition, the serving distance has a Rayleigh distribution of scale parameter $\sqrt{2}\sigma$, i.e., $f_R(r)=  \frac{r}{2 \sigma^2} {\rm e}^{\frac{-r^2}{4 \sigma^2}}$. It is worth noting that the serving distance is independent of the caching probability $b_i$. To clarify, from the thinning theorem \cite{haenggi2012stochastic}, the set of devices caching content $i$ in a given cluster forms a Gaussian \ac{PPP} $\Phi_{ci}$ whose intensity is $\lambda_{ci} = b_i\lambda_{c}(y)$. The \ac{PDF} of the distance between a randomly selected caching device from $\Phi_{ci}$ and the typical device is $f_R(r)$, which is again independent of $b_i$.

The received power at the typical device from a catering device located at $y_0$ relative to the cluster center is given by 
\begin{align}
P &= P_d  g_0 \lVert x_0+y_0\rVert^{-\alpha}= P_d  g_0 r^{-\alpha}			
\label{pwr}
\end{align}
where $P_d$ denotes the \ac{D2D} transmission power, $g_0\sim$ exp(1) is the complex Gaussian fading channel coefficient, and $\alpha > 2$ is the path loss exponent. Under this setup, the typical device sees two types of interference, namely, the intra- and inter-cluster interference. We first describe the inter-cluster interference, then the intra-cluster interference is characterized. The set of active devices in any remote cluster is denoted as $\mathcal{B}^q$, where $q$ refers to the access probability. Similarly, the set of active devices in the local cluster is denoted as $\mathcal{A}^q$. The received interference at the typical device from simultaneously active \ac{D2D} transmitters within the remote clusters is
\begin{align}
I_{\Phi_p^{!}}= \sum_{x \in \Phi_p^{!}} \sum_{y\in \mathcal{B}^q}   P_d g_{y_x}  \lVert x+y\rVert^{-\alpha}
=  \sum_{x\in \Phi_p^{!}} \sum_{y\in \mathcal{B}^q}  P_d g_{u}  u^{-\alpha}
\nonumber 
\end{align}
where $\Phi_p^{!}=\Phi_p \setminus x_0$ for ease of notation, $y$ is the marginal distance between a potential interfering device and its cluster center at $x \in \Phi_p$, $u = \lVert x+y\rVert$ is a realization of a \ac{RV} $U$ that models the inter-cluster interfering distance,  
$g_{y_x} \sim $ exp(1), and $g_{u} = g_{y_x}$. The intra-cluster interference is then given by
\begin{align}
I_{\Phi_c} &=  \sum_{y\in \mathcal{A}^p}   P_d g_{y_{x_0}}  \lVert x_0+y\rVert^{-\alpha}
=   \sum_{y\in \mathcal{A}^p}  P_d g_{h}  h^{-\alpha}
\nonumber 
\end{align}
where $y$ is the marginal distance between the intra-cluster interfering devices  and the cluster center at $x_0 \in \Phi_p$, $h = \lVert x_0+y\rVert$ is a realization of a \ac{RV} $H$, which models the intra-cluster interfering distance,  
$g_{y_{x_0}} \sim $ exp(1), and $g_{h} = g_{y_{x_0}}$. From the thinning theorem \cite{haenggi2012stochastic}, the set of active transmitters based on the slotted-ALOHA medium access forms a Gaussian \ac{PPP} $\Phi_{cq}$ whose intensity is given by
 \begin{align}
\lambda_{cq} = q\lambda_{c}(y) = q\overline{n}f_Y(y) =\frac{q\overline{n}}{2\pi\sigma^2}\textrm{exp}\Big(-\frac{\lVert y\rVert^2}{2\sigma^2}\Big),	\quad		y \in \mathbb{R}^2
\nonumber 
 \end{align}
 Assuming that the thermal noise is neglected as compared to the aggregate interference, the received 
 \ac{SIR} at the typical device can be written as 
\begin{equation}
\sir_{|r}=\textbf{1}\{N_r=1\} \frac{P}{I_{\Phi_p^{!}} + I_{\Phi_c}} = \textbf{1}\{N_r=1\} \frac{P_d  g_0 r^{-\alpha}}{I_{\Phi_p^{!}} + I_{\Phi_c}}
\end{equation}
\black{where $\textbf{1}\{.\}$ is the indicator function, and for ease of exposition, $N_r=N_{y_0}$ is a Bernoulli \ac{RV} that takes the value one with probability $q$. Thus, the event $\{N_r=1\}$ captures the incident when the serving device is admitted to access the channel.}  
Then, the probability that the received \ac{SIR} is higher than the required threshold $\vartheta$ is derived as follows:
\begin{align}
\Upsilon_{|r} = \Pb(\sir_{|r}>\vartheta)&= \Pb \Big(\textbf{1}\{ N_r=1\}\frac{P_d  g_0 r^{-\alpha}}{I_{\Phi_p^{!}} + I_{\Phi_c}} > \vartheta \Big)	
\nonumber  \\
&\overset{(a)}{=}  q   \Pb\Big(\frac{P_d  g_0 r^{-\alpha}}{I_{\Phi_p^{!}} + I_{\Phi_c}} > \vartheta\Big)
\end{align}
where (a) follows from the assumption of a Bernoulli's \ac{RV} with mean $q$. Rearranging the right-hand side, we get
\begin{align}
\Upsilon_{|r} &\overset{(b)}{=} q \mathbb{E}_{I_{\Phi_p^{!}},I_{\Phi_c}}\Big[\text{exp}\big(\frac{-\vartheta r^{\alpha}}{P_d}{[I_{\Phi_p^{!}} + I_{\Phi_c}] }\big)\Big]
\nonumber  \\
\label{prob-R1-g-R0}
 &\overset{(c)}{=}  q \mathscr{L}_{I_{\Phi_p^{!}}}(s) \mathscr{L}_{I_{\Phi_c}} (s)
\end{align}			
where (b) follows from the assumption $g_0 \sim \mathcal{CN}(0,1)$, and (c) follows from the independence of the intra- and inter-cluster interference and calculating the Laplace transform of them, with $s=\frac{\vartheta r^{\alpha}}{P_d}$. The classical tradeoff between frequency reuse and higher interference power is represented in (\ref{prob-R1-g-R0}). In other words, increasing the access probability $q$ allows more opportunities to access the channel, but this channel access would then be accompanied with higher interference power.

Next, we first derive the Laplace transform of interference to obtain the rate coverage probability $\Upsilon$. Then, we formulate the offloading gain maximization problem.
\begin{lemma}
\label{lemma-inter}
Laplace transform of the inter-cluster aggregate interference $I_{\Phi_p^{!}}$ is given by 
\begin{align}
 \label{LT_inter}
\mathscr{L}_{I_{\Phi_p^{!}}}(s) &= {\rm exp}\Big(-2\pi\lambda_p \int_{v=0}^{\infty}\Big(1 -  {\rm e}^{-q\overline{n} \varphi(s,v)}\Big)v\dd{v}\Big),
\end{align}
where $s=\frac{\vartheta r^{\alpha}}{P_d}$, $\varphi(s,v) = \int_{u=0}^{\infty}\frac{s}{s+ u^{\alpha}}f_U(u|v)\dd{u}$, and $f_U(u|v) = \mathrm{Rice} (u| v, \sigma)$ represents Rice's \ac{PDF} of parameter $\sigma$, and $v=\lVert x\rVert$.
\end{lemma}

\begin{proof}
Please see the Appendix.
\end{proof}

\begin{lemma}
\label{intra-int}
Laplace transform of the intra-cluster aggregate interference $I_{\Phi_c}$ is approximated as 
\begin{align}   
\label{LT_intra}
         \mathscr{L}_{I_{\Phi_c} }(s) \approx  {\rm exp}\Big(-q\overline{n} \int_{h=0}^{\infty}\frac{s}{s+ h^{\alpha}}f_H(h)\dd{h}\Big), 
\end{align} 
where $f_H(h) =\mathrm{Rayleigh}(h,\sqrt{2}\sigma)$ represents Rayleigh's \ac{PDF} with scale parameter $\sqrt{2}\sigma$. 
\end{lemma}
The proof of Lemma \ref{intra-int} proceeds in a similar way to the proof of Lemma \ref{lemma-inter}, and the approximation follows from neglecting the correlation among intra-cluster serving distances, i.e., the common part $x_0$ in $\lVert x_0 + y\rVert$. The proof is omitted due to limited space.

To validate the approximation in Lemma 2, in Fig.~\ref{cov_prob_vs_sigma_new}, we plot the rate coverage probability $\Upsilon$, computed from (\ref{offloading_gain}), against the displacement standard deviation $\sigma$. Fig.~\ref{cov_prob_vs_sigma_new} verifies that the adopted approximation is accurate. It is intuitive to see that the $\Upsilon$ decreases as both $\sigma$ and $\lambda_p$ increase. This is attributed to the fact that the desired signal level decreases as $\sigma$ decreases, meanwhile, the interference power increases with $\lambda_p$ and $\sigma$. 
\begin{figure}[!tb]
\vspace{-0.7 cm}
	\begin{center}
		\includegraphics[width=2.4 in]{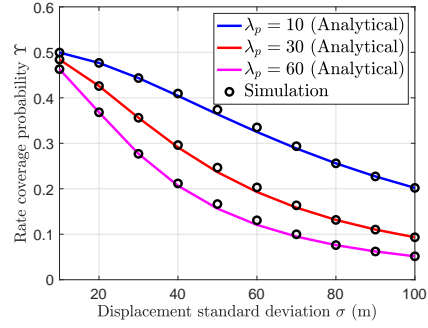}	
		\caption {The rate coverage probability $\Upsilon$ versus the displacement standard deviation $\sigma$ ($\overline{n}=5$, $\vartheta=\SI{0}{dB}$, $p=0.3$).}
		\label{cov_prob_vs_sigma_new}
	\end{center}
\vspace{-0.4 cm}
\end{figure}
From (\ref{offloading_gain}), (\ref{LT_inter}), and (\ref{LT_intra}), we get
 \begin{align}
 \Pb_o(q,\boldsymbol{b}) &= \sum_{i=1}^{N_f} p_i b_i + p_i(1 - b_i)(1 - e^{-b_i\overline{n}}) \times
 \nonumber \\
 &
 \label{offloading_gain_1}
 \quad\quad
\int_{r=0}^{\infty} 
 \frac{r}{2 \sigma^2} {\rm e}^{\frac{-r^2}{4 \sigma^2}}  p
 \mathscr{L}_{I_{\Phi_p^{!}}}(s) \mathscr{L}_{I_{\Phi_c}} (s)\dd{r},				
 \end{align}
 
 Having characterized the offloading gain, next, we formulate the joint channel access and caching optimization  problem.  
 \section{Maximizing Offloading Gain}
The offloading gain maximization problem is formulated as
\begin{align}
\label{optimize_eqn_p}
\textbf{P1:}		\quad &\underset{q,\boldsymbol{b}}{\text{max}} \quad \Pb_o(q,\boldsymbol{b}) \\
\label{const110}
&\textrm{s.t.}\quad  \sum_{i=1}^{N_f} b_i = M,   \\
\label{const111}
&  b_i \in [ 0, 1],  \\
\label{const112}
&  q \in [ 0, 1], 
\end{align}	
where (\ref{const110}) is the device cache size constraint. Since the offloading gain depends on the caching function $\boldsymbol{b}$ and the access probability $q$, and since $q$ exists as a complicated exponential term in $\Upsilon$ (see (\ref{prob-R1-g-R0}), (\ref{LT_inter}), and (\ref{LT_intra})), it is difficult to analytically characterize the objective function, e.g., show concavity or find a tractable
expression for the optimal access probability. In order to tackle this, we propose to find the optimal  access probability $q^*$ that maximizes $\Upsilon$ via the bisection search method in its feasible range $q \in [ 0, 1]$. Then, the obtained $q^*$ is used to solve for the caching probability $\boldsymbol{b}$ in the optimization problem below. 
%
%
\begin{align}
\label{optimize_eqn_b_i}
\textbf{P2:}		\quad &\underset{\boldsymbol{b}}{\text{max}} \quad \Pb_o(q^*,\boldsymbol{b}) \\
\label{const11}
&\textrm{s.t.}\quad  (\ref{const110}), (\ref{const111})   \nonumber 
\end{align}
\begin{lemma}
For fixed $q^*$, $\Pb_o(q^*,\boldsymbol{b})$ is a concave function w.r.t. $\boldsymbol{b}$ and the optimal caching probability $\boldsymbol{b}^{*}$ that maximizes the offloading gain is given by
      \[
    b_{i}^{*}=\left\{
                \begin{array}{ll}
                  1 \quad\quad\quad , v^{*} <   p_i -p_i(1-e^{-\overline{n}})\Upsilon\\ 
                  0   \quad\quad\quad, v^{*} >   p_i + \overline{n}p_i\Upsilon\\
                 \psi(v^{*}) \quad, {\rm otherwise}
                \end{array}   
              \right.
  \]
where $\psi(v^{*})$ is the solution of $v^{*} =   p_i  +  p_i\big(\overline{n}(1-b_i^*)e^{-\overline{n}b_i^*} - (1 - e^{-\overline{n}b_i^*})\big)\Upsilon$, that satisfies $\sum_{i=1}^{N_f} b_i^*=M$.
\end{lemma}
\begin{proof}
It can be easily verified that $\frac{\partial^2 \Pb_o}{\partial b_i^2}$ is always negative, and 
$\frac{\partial^2 \Pb_o}{\partial b_i \partial b_j}=0$ for all $i\neq j$. Hence, the Hessian matrix $\boldsymbol{H}_{i,j}$ of $\Pb_o(q^*,\boldsymbol{b})$ w.r.t. $\boldsymbol{b}$ is negative semidefinite, and $\Pb_o(q^*,\boldsymbol{b})$ is a concave function of $\boldsymbol{b}$. Also, the constraints are linear, which implies that the necessary and sufficient conditions for optimality exist. The dual Lagrangian function and the \ac{KKT} conditions can be employed to solve \textbf{P2}, with the details omitted due to the limited space. 
\end{proof}

Clearly, the optimal caching solution $\boldsymbol{b}^*$ depends on the scheduling of devices through channel access probability $q^*$ from $\Upsilon$, while $q^*$ is independent of $\boldsymbol{b}^*$. \cite{cache_schedule} shows that a \ac{PPP} network exhibits the same property, i.e., the caching scheme is scheduling-dependent. To gain some insights, it is useful to consider a simple case when only one D2D link per cluster is allowed. In this case, the rate coverage probability of the proposed clustered model with one active \ac{D2D} link within a cluster will be \cite[Lemma 2]{8647532}:   
\begin{align}
\Upsilon = \frac{1}{\big(4\sigma^2\pi\lambda_p \vartheta^{2/\alpha}\Gamma(1 + 2/\alpha)\Gamma(1 - 2/\alpha)+ 1\big)}. 
\end{align}
%
Substituting in (\ref{offloading_gain_1}) for $\Upsilon$, we get the offloading gain as
\begin{align}
 \label{offloading_gain_geometry}
 \Pb_o(\boldsymbol{b}) &= \sum_{i=1}^{N_f} p_i b_i + \frac{p_i(1 - b_i)(1 - e^{-b_i\overline{n}})}{4\sigma^2\pi\lambda_p \vartheta^{2/\alpha}\Gamma(1 + 2/\alpha)\Gamma(1 - 2/\alpha)+ 1} ,				
 \end{align}
\begin{remark} 
\label{remark1}
{\rm From (\ref{offloading_gain_geometry}), it is clear that the offloading gain increases as $\sigma$ and $\lambda_p$ decrease. Particularly, the offloading gain is inversely proportional to the density of clusters $\lambda_p$ and the variance of the displacement $\sigma^2$. This is because smaller $\sigma$ results in higher levels of the desired signal, while lower $\lambda_p$ leads to smaller encountered interference at the typical device.} 
\end{remark}
\vspace{-0.5 cm}
\section{Numerical Results}
\begin{table}[!tb]
\vspace{-0.5 cm}
\caption{Simulation Parameters} 
\centering 
\begin{tabular}{l l l} 
\hline\hline 
Description & Parameter & Value  \\ [0.5ex] 
\hline 
Displacement standard deviation & $\sigma$ & \SI{10}{\metre} \\ 
Popularity index&$\beta$&0.5\\
Path loss exponent&$\alpha$&4\\
Library size and cache size per device&$N_f$, $M$&100, 8 files\\
\black{Average number of devices per cluster}&$\overline{n}$&4\\
Density of clusters&$\lambda_{p}$&10 clusters/\SI{}{km}$^2$\\
$\sir$ threshold&$\vartheta$&\SI{0}{\deci\bel}\\
\hline 
\end{tabular}
\label{ch3:table:sim-parameter} 
\end{table}
We first validate the developed mathematical model via Monte Carlo simulations. Then we benchmark the proposed caching scheme against conventional caching schemes. Unless otherwise stated, the network parameters are selected as shown in Table \ref{ch3:table:sim-parameter}. 
\begin{figure}[!tbh]			
\vspace{-0.4 cm}
	\begin{center}
		\includegraphics[width=2.2 in]{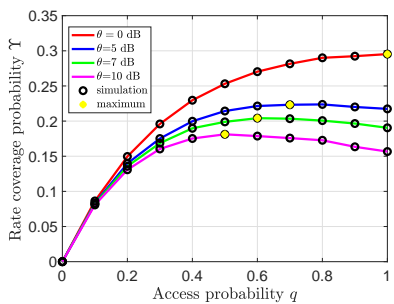}	
		\caption{The rate coverage probability $\Upsilon$ versus the access probability $q$.}
		\label{prob_r_geq_r0_vs_p}
	\end{center}
\vspace{-0.3 cm}
\end{figure}		

In Fig.~\ref{prob_r_geq_r0_vs_p}, we plot the rate coverage probability $\Upsilon$ against the channel access probability $q$. The theoretical and simulated results are plotted together, and they are consistent. Clearly, there is an optimal $q^*$; before it $\Upsilon$ tends to increase as the probability of accessing the channel increases, and beyond it, $\Upsilon$ tends to decrease due to the effect of aggressive interference. It is intuitive to observe that the optimal access probability $q^*$, which maximizes $\Upsilon$, decreases as $\vartheta$ increases. This reflects the fact the system becomes more sensitive to the effect of interference when a higher \ac{SIR} threshold is required.


\begin{figure}[!t]			
\vspace{-1.1 cm}
	\begin{center}
		\includegraphics[width=2.2 in]{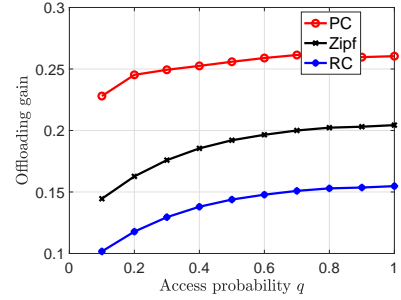}			
		\caption {The offloading gain versus the access probability $q$.}			
		\label{offloading_gain_vs_scaling}
	\end{center}
\vspace{-0.0 cm}
\end{figure}
Fig.~\ref{offloading_gain_vs_scaling} manifests the effect of the access probability $q$ on the offloading gain. The offloading gain is plotted against $q$ for different caching schemes, namely, the proposed \ac{PC}, Zipf caching (Zipf), and uniform \ac{RC}. Fig.~\ref{offloading_gain_vs_scaling} is plotted for an \ac{SIR} threshold $\vartheta=\SI{0}{\deci\bel}$, hence, the optimal access probability $q^*$ is near one from Fig.~{\ref{prob_r_geq_r0_vs_p}}. Clearly, the offloading gain for the different caching schemes improves as $q$ approaches its optimal value, which reveals the crucial impact of the device scheduling on the content placement and accordingly, on the offloading gain. Moreover, the proposed \ac{PC} is shown to attain the best performance as compared to other benchmark schemes. 
 
\begin{figure}[!t]
\vspace{-0.2 cm}
\centering
  \subfigure[$q=q^*$. ]{\includegraphics[width=1.65 in]{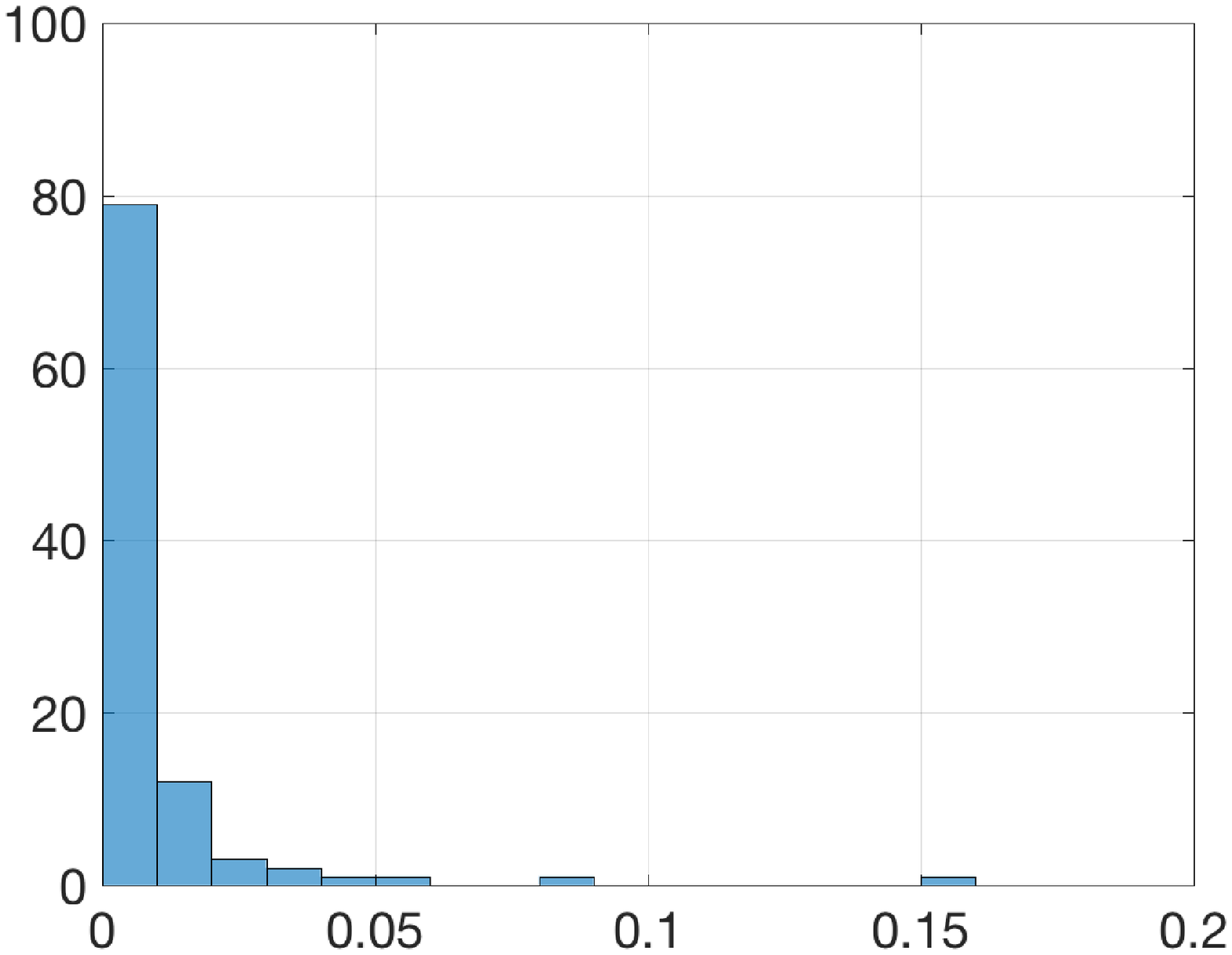}		
\label{histogram_b_i_p_star}}
\subfigure[$q < q^*$.]{\includegraphics[width=1.65 in]{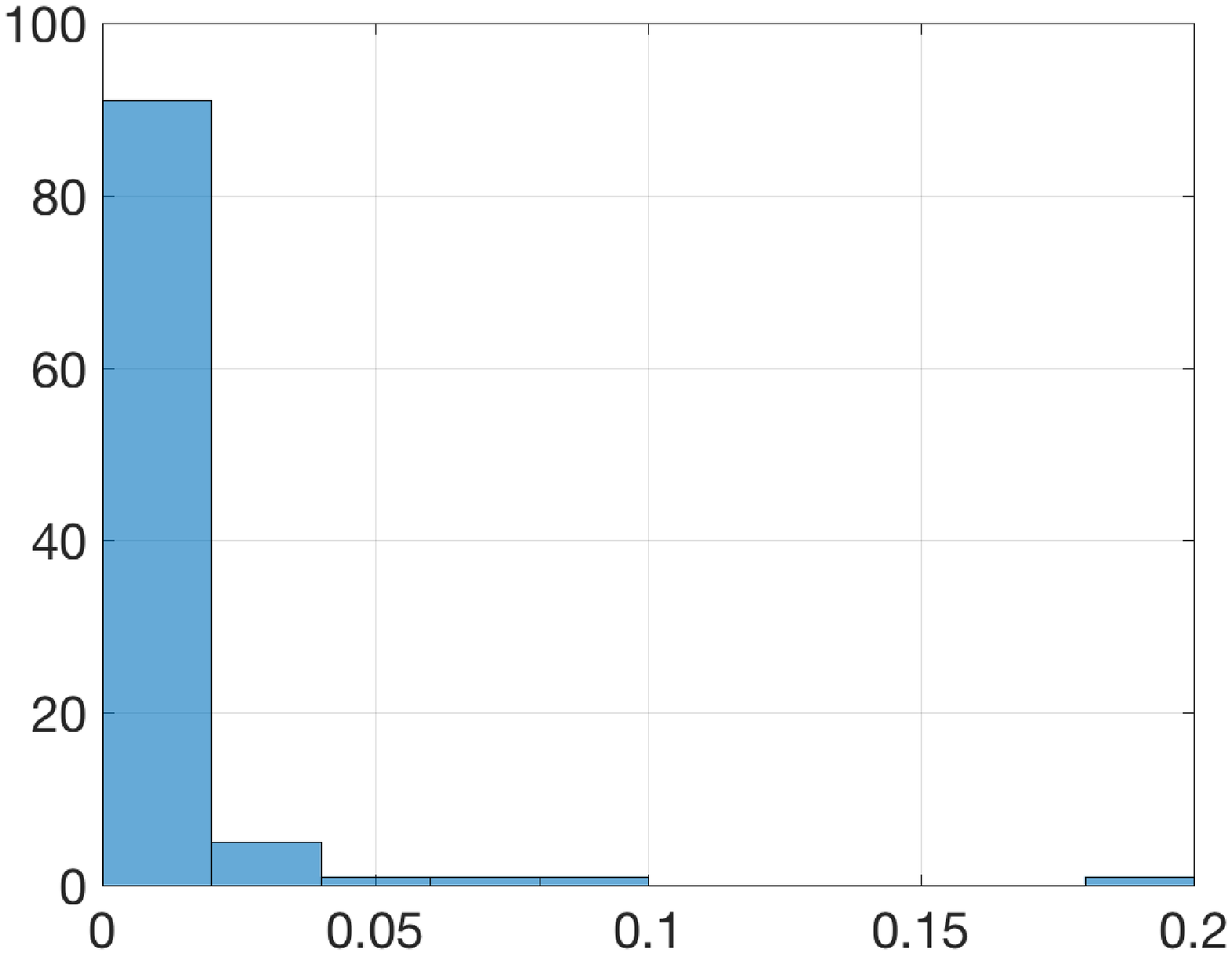}		
  \label{histogram_b_i_p_leq_p_star}}
\caption{Histogram of the optimal caching probability $\boldsymbol{b}^*$.}
\label{histogram_b_i}
\vspace{-0.3 cm}
\end{figure}
\black{To show the effect of $q$ on the caching probability, in Fig.~\ref{histogram_b_i}, we plot the histogram of the optimal caching probability at different $q$ values.  Specifically, $q=q^*$ in Fig.~\ref{histogram_b_i_p_star} and $q<q^*$ in Fig.~\ref{histogram_b_i_p_leq_p_star}. It is clear from the histograms that the optimal caching probability $\boldsymbol{b}^*$ tends to be more skewed when $q < q^*$, i.e., when $\Upsilon$ decreases. This shows that file sharing is more difficult when $q$ is not optimized.
Broadly speaking, for $q<q^*$, the system is too conservative, while for $q>q^*$, the outage probability is high due to the aggressive interference. In such regimes, each device tends to cache the most popular files leading to fewer opportunities of content transfer.}

\begin{figure}[!t]
\vspace{-1.2 cm}
	\begin{center}
		\includegraphics[width=2.2 in]{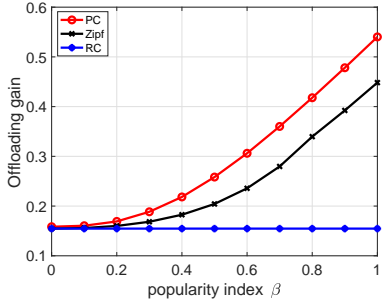}			
		\caption {The offloading gain versus the popularity of files $\beta$.}	
		\label{offloading_gain_vs_beta}
	\end{center}
\vspace{-0.2 cm}
\end{figure}
Fig.~\ref{offloading_gain_vs_beta} illustrates the prominent effect of the content popularity on the offloading gain, and compares the achievable gain of the above mentioned caching schemes. Clearly, the offloading gain of the proposed \ac{PC} attains the best performance as compared to other schemes. Particularly, 10$\%$ improvement in the offloading gain is observed compared to the Zipf caching when $\beta=1$. Moreover, we note that all caching schemes encompass the same offloading gain when $\beta=0$ owing to the uniformity of content popularity.

To show the effect of network geometry, in Fig.~\ref{offloading_gain_geometry-fig}, we plot the closed-form  offloading gain in (\ref{offloading_gain_geometry}) against $\sigma$ at different $\lambda_p$. Fig.~\ref{offloading_gain_geometry-fig} shows that the offloading gain monotonically decreases with both $\sigma$ and $\lambda_p$. This is because content sharing between devices turns out to be less successful when the distance between devices is large, i.e., larger $\sigma$. \black{This result is also aligned with the outcome of Fig.~\ref{cov_prob_vs_sigma_new} which showed that the rate coverage probability $\Upsilon$ decreases as $\sigma$ or $\lambda_p$ increase}. Analogously, file sharing among the cluster devices is accompanied with higher interference when $\lambda_p$ and $\sigma$ are higher. Accordingly, this expected degradation results in less successful content delivery via \ac{D2D} communication. 
\begin{figure}[!t]
\vspace{-0.50 cm}
	\begin{center}
		\includegraphics[width=2.2 in]{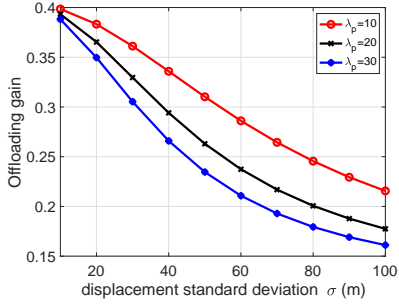}			
		\caption {The offloading gain versus the displacement standard deviation $\sigma$ at different cluster densities $\lambda_p$.}
		\label{offloading_gain_geometry-fig}
	\end{center}
\vspace{-0.7 cm}
\end{figure}

\vspace{-0.2 cm}
\section{Conclusion}
\vspace{-0.0 cm}
In this paper, we have proposed a joint communication and caching optimization framework for clustered D2D networks. In particular, we have conducted joint optimization of channel access probability and content placement in order to maximize the offloading gain. We have characterized the optimal content caching scheme as a function of the system parameters, namely, density of clusters, average number of devices per cluster, caching scheme, and access probabilities. A bisection search method is also proposed to calculate the optimal channel access probability. We have demonstrated that deviating from the optimal access probability makes file sharing more difficult, i.e., the system is too conservative for small access probabilities, while the interference is too aggressive for larger access probabilities. Results showed up to $10\%$ enhancement in offloading gain for the proposed approach compared to the Zipf caching technique.\footnote{Creating communication protocols for secure content delivery for networks
of UAVs using, e.g., blockchain technology, can be a potential subject for
future investigation \cite{baza5,baza2,baza3,baza6,baza4,baza1,baza7,baza13,baza9,baza11,baza10,baza12,baza8}, 
\cite{8756296,amer2020mobiledrone,8998329},
 \cite{7440654,7605066,7925123}.	} 
\vspace{-0.2 cm}
\begin{appendix}
Laplace transform of the inter-cluster aggregate interference $I_{\Phi_p^{!}}$ can be evaluated as
\begin{align}
&\mathscr{L}_{I_{\Phi_p^{!}}}(s)= \mathbb{E} \Bigg[e^{-s \sum_{\Phi_p^{!}} \sum_{y \in \mathcal{B}^q}  g_{y_{x}}  \lVert x + y\rVert^{-\alpha}} \Bigg] 
  \nonumber \\
    &
    \overset{(a)}{=} \mathbb{E}_{\Phi_p} \Bigg[\prod_{\Phi_p^{!}} \mathbb{E}_{\Phi_{cq}} \prod_{y \in \mathcal{B}^q} \frac{1}{1+s \lVert x + y\rVert^{-\alpha}} \Bigg]  \nonumber \\
     &\overset{(b)}{=} \mathbb{E}_{\Phi_p} \prod_{\Phi_p^{!}} e^{-q\overline{n} \int_{\mathbb{R}^2}\big(1 - \frac{1}{1+s \lVert x + y\rVert^{-\alpha}}\big)f_Y(y)\dd{y}} 	 \nonumber	
\end{align}
\begin{align}
          &\overset{(c)}{=}  e^{-\lambda_p \int_{\mathbb{R}^2}\big(1 -  e^{-q\overline{n} \int_{\mathbb{R}^2}\big(1 - \frac{1}{1+s \lVert x + y\rVert^{-\alpha}}\big)f_Y(y)\dd{y}}\dd{x}  }             	
\nonumber
\end{align}
where (a) follows from the Rayleigh fading assumption, (b) follows from the \ac{PGFL} of Gaussian \ac{PPP} $\Phi_{cq}$, and (c) follows from the \ac{PGFL} of the parent \ac{PPP} $\Phi_p$. By using change of variables $z = x + y$ with $\dd z = \dd y$, we proceed as
\begin{align}			
& \mathscr{L}_{I_{\Phi_p^{!}}}(s)\overset{}{=}  e^{-\lambda_p \int_{\mathbb{R}^2}\Big(1 -  e^{-q\overline{n} \int_{\mathbb{R}^2}\big(1 - \frac{1}{1+s \lVert z\rVert^{-\alpha}}\big)f_Y(z-x)\dd{y}}\Big)\dd{x}}
\end{align}			 
\begin{align}
          &\overset{(d)}{=}  e^{-2\pi\lambda_p \int_{v=0}^{\infty}\Big(1 -  e^{-q\overline{n} \int_{u=0}^{\infty}\big(1 - \frac{1}{1+s  u^{-\alpha}}\big)f_U(u|v)\dd{u}}\Big)v\dd{v}}		 
          \nonumber		\\
         &=  e^{-2\pi\lambda_p \int_{v=0}^{\infty}\Big(1 -  e^{-q\overline{n} \int_{u=0}^{\infty}
         \frac{s}{s+ u^{\alpha}}f_U(u|v)\dd{u}}\Big)v\dd{v}},		 
\end{align}
where (d) follows from converting the cartesian coordinates to the polar coordinates with $u=\lVert z\rVert$. 
\black{To clarify how in (d) the normal distribution $f_Y(z-x)$ is converted to the Rice distribution $f_U(u|v)$, consider a remote cluster centered at $x \in \Phi_p^!$, with a distance $v=\lVert x\rVert$ from the origin. Every interfering device belonging to the cluster centered at $x$ has its coordinates in $\mathbb{R}^2$ chosen independently from a Gaussian distribution with standard deviation $\sigma$. Then, by definition, the distance from such an interfering device to the origin, denoted as $u$, has a Rice distribution, denoted as $f_U(u|v)=\frac{u}{\sigma^2}\mathrm{exp}\big(-\frac{u^2 + v^2}{2\sigma^2}\big) I_0\big(\frac{uv}{\sigma^2}\big)$, where $I_0$ is the modified Bessel function of the first kind with order zero and $\sigma$ is the scale parameter.
Letting $\varphi(s,v) = \int_{u=0}^{\infty}\frac{s}{s+ u^{\alpha}}f_U(u|v)\dd{u}$, the proof is completed.}

\end{appendix}
\bibliographystyle{IEEEtran}

\vspace{-0.2 cm}
\bibliography{public_files/bibliography}

\end{document}

%% file: public_files/mypackage.tex
\usepackage[utf8]{inputenc}
\usepackage[T1]{fontenc}
\usepackage{url}
\usepackage{ifthen}
\usepackage{cite}
\usepackage[cmex10]{amsmath} 
\usepackage{stackengine}
\usepackage{gensymb}
\usepackage{graphics}
\usepackage{mypackage}
\usepackage{mathrsfs}

\usepackage{mathtools}
\usepackage{tabu}
\usepackage{float}
\usepackage{physics}
\usepackage{siunitx}
\usepackage{multicol}
\usepackage{amssymb}
\usepackage[nomain,acronym,shortcuts]{glossaries}

\usepackage{etoolbox}

%% file: public_files/newcommands.tex
\newcommand{\sir}{\mathrm{SIR}}
\newcommand{\Pb}{\mathbb{P}}

 \makeglossaries
\newcommand*{\acro}[3][]{\newacronym[#1]{#2}{#2}{#3}}

\newtheorem{remark}{Remark}

\theoremstyle{approximation}

\newcommand{\black}{\textcolor{black}}

\let\mybibitem\bibitem
\renewcommand{\bibitem}[1]{%
  \ifstrequal{#1}{nature}
    {\color{blue}\mybibitem{#1}}
    {\color{black}\mybibitem{#1}}%
}


%% file: public_files/acronyms.tex
\acro{OFDM}{orthogonal frequency-division multiplexing}
\acro{OFDMA}{orthogonal frequency-division multiple access}
\acro{MNO}{mobile network operator}
\acro{RA}{resource allocation}
\acro{SC-FDMA}{single carrier frequency division multiple access}
\acro{CR}{cognitive radio}
\acro{RFIC}{radio frequency integrated circuit} 
\acro{SDR}{software defined radio}
\acro{SDN}{software defined networking}
\acro{su}{secondary user}
\acro{QoS}{quality-of-service}
\acro{USRP}{universal software radio peripheral}
\acro{GSM}{Global System for Mobile Communications}
\acro{TDMA}{time-division multiple access}
\acro{FDMA}{frequency-division multiple access}
\acro{GPRS}{General Packet Radio Service}
\acro{MSC}{Mobile Switching Centre}
\acro{BSC}{Base Station Controller}
\acro{UMTS}{Universal Mobile Telecommunications System}
\acro{WCDMA}{wide-band code division multiple access}
\acro{CDMA}{code division multiple access}
\acro{LTE}{Long Term Evolution}
\acro{PAPR}{peak-to-average power rating}
\acro{HetNet}{heterogeneous networks}
\acro{PHY}{physical layer}
\acro{MAC}{medium access control}
\acro{AMC}{adaptive modulation and coding}
\acro{MIMO}{multiple-input multiple-output}
\acro{M-MIMO}{massive MIMO}
\acro{RAT}{radio access technology}
\acro{VNI}{visual networking index}
\acro{RB}{resource block}
\acro{UE}{user equipment}
\acro{CQI}{channel quality indicator}
\acro{HD}{half duplex}
\acro{IBFD}{in-band full duplex}
\acro{SIC}{self-interference cancellation}
\acro{SI}{self-interference}
\acro{BS}{base station}
\acro{FBMC}{filter bank multi-carrier}
\acro{UFMC}{universal filtered multi-carrier}
\acro{SCM}{single carrier modulation}
\acro{isi}{inter-symbol interference}
\acro{FTN}{faster-than-nyquist}
\acro{M2M}{machine-to-machine}
\acro{MTC}{machine type communication}
\acro{mmWave}{millimeter wave}
\acro{BF}{beamforming}
\acro{LoS}{line-of-sight}
\acro{NLoS}{non-line-of-sight}
\acro{CAPEX}{capital expenditure}
\acro{OPEX}{operational expenditure}
\acro{ICT}{information and communications technology}
\acro{SP}{service providers}
\acro{InP}{infrastructure providers}
\acro{MVNP}{mobile virtual network provider}
\acro{MVNO}{mobile virtual network operator}
\acro{NFV}{network function virtualization}
\acro{VNF}{virtual network functions}
\acro{C-RAN}{cloud radio access network}
\acro{RAN}{radio access network}
\acro{BBU}{baseband unit}
\acro{RRH}{remote radio head}
\acro{TDD}{time-division duplexing}
\acro{FDD}{frequency-division duplexing}
\acro{GFDM}{generalized frequency division multiplexing}
\acro{CSI}{channel state information}
\acro{FFT}{fast Fourier transform}
\acro{IFFT}{inverse FFT}
\acro{CFO}{carrier frequency offset}
\acro{CoMP}{coordinated multipoint}
\acro{D2D}{device-to-device}
\acro{OOB}{out-of-band}
\acro{TTI}{transmission time interval}
\acro{DUE}{D2D user equipment}
\acro{DAS}{distributed antenna system}
\acro{ICIC}{inter-cell interference coordination}
\acro{PDF}{probability distribution function}
\acro{KPI}{key performance indicator}
\acro{SBS}{small base station}
\acro{MBS}{macro base station}
\acro{SCN}{small cell network}
\acro{FIFO}{first in first out}
\acro{VCC}{virtual cache center}
\acro{UAV}{unmanned aerial vehicles}
\acro{MPSQ}{multiclass processor sharing queue} 
\acro{EE}{energy efficiency}
\acro{SIR}{signal-to-interference ratio}
\acro{SINR}{signal-to-noise-plus-interference ratio}
\acro{PPP}{Poisson point process}
\acro{PCP}{Poisson cluster process}
\acro{HCP}{hard-core placement} 
\acro{TCP}{Thomas cluster process }
\acro{CPF}{caching popular files}
\acro{GCA}{greedy caching algorithm }
\acro{RC}{random caching }
\acro{PC}{probabilistic caching}
\acro{5G}{fifth generation}
\acro{MEC}{mobile edge computing}
\acro{AP}{access point}
\acro{VoD}{video-on-demand}
\acro{EPC}{evolved packet core}
\acro{QoE}{quality-of-experience}
\acro{CDN}{content delivery networks}
\acro{F-RAN}{fog-radio access network}
\acro{AR}{augmented reality}
\acro{VR}{virtual reality}
\acro{4C}{computing, caching, communication, and control}
\acro{ABR}{Adaptive BitRate}
\acro{ILP}{Integer Linear Program}
\acro{MILP}{Mixed Integer Linear Program}
\acro{MINLP}{Mixed Integer Non-Linear Program}
\acro{BC}{broadcast channel}
\acro{MDS}{maximum distance separable}
\acro{RS}{Reed-Solomon}
\acro{NC}{network coding}
\acro{MSR}{minimum storage regenerating}
\acro{Coop-MIMO}{cooperative \ac{mimo}}
\acro{UDN}{ultra-dense networks}
\acro{ETSI}{European Telecommunications Standards Institute}
\acro{MCP}{Matern cluster process}
\acro{MD-CoMP}{macrodiversity CoMP transmission}
\acro{JT-CoMP}{joint transmission CoMP}
\acro{CoMP-JT}{coordinated multipoint joint transmission}
\acro{MDSD}{multiple devices to the single device}
\acro{PMF}{probability mass function}
\acro{RV}{random variable}
\acro{i.i.d.}{independently and identically distributed}
\acro{MBMS}{multimedia broadcasting multicasting service}
\acro{CCDF}{complementary cumulative distribution function}
\acro{CDF}{cumulative distribution function}
\acro{PGFL}{probability generating functional}
\acro{KKT}{Karush-Kuhn-Tucker}
\acro{PGF}{point generating function}
\acro{BPP}{binomial point process}
\acro{3GPP}{3rd Generation Partnership Project}
\acro{FAA}{Federal Aviation Administration}
\acro{UAS}{unmanned aircraft systems}
\acro{IoT}{Internet-of-things}
\acro{CNPC}{control and non-payload communication}
\acro{FHD}{full high definition}
\acro{MGF}{moment generating function}
\acro{3D}{three-dimensional}
\acro{2D}{two-dimensional}
\acro{1D}{one-dimensional}
\acro{CB}{conjugate beamforming}
\acro{AU}{aerial user}
\acro{GU}{ground user}
\acro{CLT}{central limit theorem}
\acro{SE}{spectral efficiency}
\acro{MRT}{maximum ratio transmission}
\acro{ZFBF}{zero-forcing beamforming}
\acro{GBS}{ground base station}
\acro{C&C}{command and control}
\acro{ASE}{area spectral efficiency}
\acro{RWP}{random waypoint}
\acro{AGL}{above ground level}
\acro{w.r.t.}{with respect to}
\acro{GTA}{ground-to-air}
\acro{AUE}{aerial user equipment}
\acro{GUE}{ground user equipment}
\acro{UB}{upper bound}
\acro{LB}{lower bound}
\acro{ABS}{aerial base station}
\acro{UAV-UE}{UAV-\ac{UE}}
\acro{SCDP}{successful content delivery probability}
\acro{HARP}{highest average received power}
\acro{ULA}{uniform linear array}
\acro{GPP}{Gaussian Poisson process}
\acro{NSD}{nearest serving device}
\acro{NCP}{nearest content provider}
\acro{RSD}{randomly serving device}
\acro{RSCP}{randomly-selected content provider}
\acro{BCD}{block coordinate descent}
\acro{NOMA}{non-orthogonal multiple access}
\acro{MPC}{most popular content}
\acro{LPC}{least popular content}
\acro{AI}{artificial intelligence}
\acro{ML}{machine learning}
\acro{TPP}{temporal point process}
\acro{PP}{point process}
\acro{DPP}{Determinantal point process}
\acro{DPPL}{DPP-based learning}
\acro{DRL}{deep reinforcement learning}
\acro{RL}{reinforcement learning}
\acro{RNN}{recurrent neural network}
\acro{RMTPP}{recurrent marked TPP}
\acro{NR}{New Radio}
\acro{HO}{handover}
\acro{RLF}{radio link failure}
\acro{RSRP}{reference signal received power}